\documentclass[12pt]{article}%
\usepackage{mitpress}
\usepackage{amsmath}
\usepackage{graphicx}%
\usepackage{amsfonts}%
\usepackage{amssymb}
\newtheorem{theorem}{Theorem}

\newtheorem{definition}[theorem]{Definition}
\newtheorem{example}[theorem]{Example}

\newtheorem{lemma}[theorem]{Lemma}
\newtheorem{notation}[theorem]{Notation}

\newtheorem{remark}[theorem]{Remark}

\newenvironment{proof}[1][Proof]{\textbf{#1.} }{\ \rule{0.5em}{0.5em}}
\newdimen\dummy
\dummy=\oddsidemargin
\begin{document}

\title{The decomposition of the regular asynchronous systems as parallel connection
of regular asynchronous systems}
\author{Serban E. Vlad\\Str. Zimbrului, Nr. 3, Bl. PB68, Ap. 11, 410430, Oradea, Romania, email: serban\_e\_vlad@yahoo.com}
\maketitle

\begin{abstract}
The asynchronous systems are the non-deterministic models of the asynchronous
circuits from the digital electrical engineering, where non-determinism is a
consequence of the fact that modelling is made in the presence of unknown and
variable parameters. Such a system is a multi-valued function $f$ that assigns
to an (admissible) input $u:\mathbf{R}\rightarrow\{0,1\}^{m}$ a set $f(u)$ of
(possible) states $x:\mathbf{R}\rightarrow\{0,1\}^{n}.$ When this assignment
is defined by making use of a so-called generator function $\Phi
:\{0,1\}^{n}\times\{0,1\}^{m}\rightarrow\{0,1\}^{n},$ then the asynchronous
system $f$ is called regular. The generator function $\Phi$ acts in this
asynchronous framework similarly with the next state function from a
synchronous framework. The parallel connection of the asynchronous systems
$f^{\prime}$ and $f^{\prime\prime}$ is the asynchronous system $(f^{\prime
}||f^{\prime\prime})(u)=f^{\prime}(u)\times f^{\prime\prime}(u).$ The purpose
of the paper is to give the circumstances under which a regular asynchronous
system $f$ may be written as a parallel connection of regular asynchronous systems.

\end{abstract}

\section{Introduction}

The theory of modeling the asynchronous circuits from the digital electrical
engineering has its origin in the switching (circuits) theory of the 50's and
the 60's. By that time, researchers used the discrete time and mathematics
seemed to be uncensored. After 1970, instead of switching theory (more
exactly: instead of what we understand by switching theory), the analysis of
these circuits is made in general by engineers that give approximate
descriptions of the switching phenomena and draw pictures instead of writing
equations. In fact the implicit suggestion given by the literature is that the
real research is unpublished and a continuation of the switching theory,
interrupted 40 years ago, is necessary. We have tried to do so and we have
called this attempt the asynchronous systems theory.

The asynchronous systems theory makes use of $\mathbf{R}\rightarrow\{0,1\}$
functions (real time, binary values) that are not studied at all in literature
(and have never been) as far as we know. The 'nice' $\mathbf{R}\rightarrow
\{0,1\}$ functions are called signals, making us think of the digital
electrical signals.

The modeling of the asynchronous circuits is made in the presence of unknown
and variable parameters: the tension of the mains, the temperature, the delays
that depend on technology. We could have used for this reason three valued
signals $\mathbf{R}\rightarrow\{0,1,2\}$, but we prefer the present frame
because of its algebraical advantages, the set $\{0,1\}$ is organized as a
Boole algebra and as a field, unlike the set $\{0,1,2\}.$ A certain price must
be paid however, determinism is replaced by non-determinism, meaning that the
systems working with $\mathbf{R}\rightarrow\{0,1,2\}$ functions may be
considered to be input-output functions, but the systems working with
$\mathbf{R}\rightarrow\{0,1\}$ functions must be considered input-output
multi-valued functions. Thus an input is a function $u:\mathbf{R}%
\rightarrow\{0,1\}^{m},$ representing the cause and a state=output is a
function $x:\mathbf{R}\rightarrow\{0,1\}^{n},$ representing the effect. The
input is subject to certain constraints, not all the signals $\mathbf{R}%
\rightarrow\{0,1\}^{m}$ are allowed to be causes and this is why it is called
admissible. The state is not unique, the system $f$ assigns to $u$ a family
$f(u)$ of states $x\in f(u)$ that are called possible states; the unknown and
the variable parameters that accompany the switching phenomena and $f$ give us
the certitude that when $u$ is applied, an element of $f(u)$ will result.

The meaning of regularity is that of giving a special case of system, i.e.
input-output multi-valued function, when a certain circuit is really modeled,
since input-output multi-valued functions exist that model nothing. The
regular asynchronous systems are these systems $f$ that are 'generated' by a
so-called 'generator function' $\Phi:\{0,1\}^{n}\times\{0,1\}^{m}%
\rightarrow\{0,1\}^{n}$ (or 'network function' as Grigore Moisil called it;
for Moisil a 'network' is a circuit identified with its model, the system).

The parallel connection of two systems $f^{\prime}$ and $f^{\prime\prime}$ is
the system $(f^{\prime}||f^{\prime\prime})(u)=f^{\prime}(u)\times
f^{\prime\prime}(u)$ in which $f^{\prime}$ and $f^{\prime\prime}$ act
independently on each other, but under a common input $u$. When $f^{\prime}$
is generated by $\Phi^{\prime}$ and $f^{\prime\prime}$ is generated by
$\Phi^{\prime\prime},f^{\prime}||f^{\prime\prime}$ is a regular system also,
which is generated by a function denoted by $\Phi^{\prime}||\Phi^{\prime
\prime}.$ In $\Phi^{\prime}||\Phi^{\prime\prime}$ some coordinates do not
depend on other coordinates. This suggests the idea of considering the
converse situation, when, from the fact that a system $f$ is generated by the
function $\Phi$ and in $\Phi$ some coordinates do not depend on other
coordinates, we can infer the existence of some systems $f^{\prime}%
,f^{\prime\prime}$ and of some functions $\Phi^{\prime}:\{0,1\}^{n^{\prime}%
}\times\{0,1\}^{m}\rightarrow\{0,1\}^{n^{\prime}},\Phi^{\prime\prime
}:\{0,1\}^{n^{\prime\prime}}\times\{0,1\}^{m}\rightarrow\{0,1\}^{n^{\prime
\prime}},$ such that $n^{\prime}+n^{\prime\prime}=n,$ $f^{\prime}$ is
generated by $\Phi^{\prime},f^{\prime\prime}$ is generated by $\Phi
^{\prime\prime},$ $\Phi=\Phi^{\prime}||\Phi^{\prime\prime}$ and $f=f^{\prime
}||f^{\prime\prime}.$ Our present purpose is to study this possibility.

\section{Preliminaries}

\begin{definition}
The set $\mathbf{B}=\{0,1\}$ endowed with the usual algebraical laws
$\overline{\;\;},$ $\cup,\;\cdot,\;\oplus\;$and with the order $0<1$ is called
the \textbf{binary Boole algebra}.
\end{definition}

\begin{definition}
The \textbf{characteristic function} $\chi_{A}:\mathbf{R}\rightarrow
\mathbf{B}$ of the set $A\subset\mathbf{R}$ is defined by $\forall t\in A,$%
\[
\chi_{A}(t)=\left\{
\begin{array}
[c]{c}%
1,t\in A\\
0,t\notin A
\end{array}
\right.  .
\]
\end{definition}

\begin{notation}
We denote by $Seq$ the set of the sequences $t_{k}\in\mathbf{R},$
$k\in\mathbf{N}$ which are strictly increasing $t_{0}<t_{1}<t_{2}<...$ and
unbounded from above. The elements of $Seq$ will be denoted in general by
$(t_{k}).$
\end{notation}

\begin{definition}
The \textbf{signals} (or the $n-$\textbf{signals}) are by definition the
$\mathbf{R}\rightarrow\mathbf{B}^{n}$ functions of the form%
\begin{equation}
x(t)=\mu\cdot\chi_{(-\infty,t_{0})}(t)\oplus x(t_{0})\cdot\chi_{\lbrack
t_{0},t_{1})}(t)\oplus...\oplus x(t_{k})\cdot\chi_{\lbrack t_{k},t_{k+1}%
)}(t)\oplus... \label{pre1}%
\end{equation}
where $t\in\mathbf{R},$ $\mu\in\mathbf{B}^{n}$ and $(t_{k})\in Seq.$ The set
of the signals is denoted by $S^{(n)}.$
\end{definition}

\begin{definition}
In (\ref{pre1}), $\mu$ is called the \textbf{initial value} of $x$ and its
usual notation is $x(-\infty+0)$.
\end{definition}

\begin{definition}
\label{Def6}The \textbf{Cartesian product} of the functions $x^{\prime
}:\mathbf{R}\rightarrow\mathbf{B}^{n^{\prime}}$ and $x^{\prime\prime
}:\mathbf{R}\rightarrow\mathbf{B}^{n^{\prime\prime}}$ is the $\mathbf{R}%
\rightarrow\mathbf{B}^{n^{\prime}+n^{\prime\prime}}$ function denoted by
$(x^{\prime},x^{\prime\prime}),$ $(x^{\prime}(t),x^{\prime\prime}(t)),$
$x^{\prime}\times x^{\prime\prime}$ or $x^{\prime}(t)\times x^{\prime\prime
}(t)$ which is defined by $\forall i\in\{1,...,n^{\prime}+n^{\prime\prime
}\},\forall t\in\mathbf{R},$%
\[
(x^{\prime}(t),x^{\prime\prime}(t))_{i}=\left\{
\begin{array}
[c]{c}%
x_{i}^{\prime}(t),i\in\{1,...,n^{\prime}\},\\
x_{i}^{\prime\prime}(t),i\in\{n^{\prime}+1,n^{\prime}+n^{\prime\prime}\}
\end{array}
\right.  .
\]
\end{definition}

\begin{notation}
For any set $M$, we denote with $P^{\ast}(M)$ the set of the non-empty subsets
of $M.$
\end{notation}

\begin{definition}
\label{Def5}If $X^{\prime}\in P^{\ast}(S^{(n^{\prime})})$ and $X^{\prime
\prime}\in P^{\ast}(S^{(n^{\prime\prime})}),$ then their \textbf{Cartesian
product} $X^{\prime}\times X^{\prime\prime}\in P^{\ast}(S^{(n^{\prime
}+n^{\prime\prime})})$ is defined by%
\[
X^{\prime}\times X^{\prime\prime}=\{(x^{\prime},x^{\prime\prime})|x^{\prime
}\in X^{\prime},x^{\prime\prime}\in X^{\prime\prime}\}.
\]
\end{definition}

\begin{remark}
The signals model the electrical signals from the digital electrical
engineering and $\mathbf{R}$ is the time set. The last two definitions of the
Cartesian product replace the usual definition of $x^{\prime}\times
x^{\prime\prime}$ as a $\mathbf{R}\times\mathbf{R}\rightarrow\mathbf{B}%
^{n^{\prime}}\times\mathbf{B}^{n^{\prime\prime}}$ function with $x^{\prime
}\times x^{\prime\prime}:\mathbf{R}\rightarrow\mathbf{B}^{n^{\prime}%
+n^{\prime\prime}}$ because the time is unique. We have suggested this idea in
the notation of $S^{(n)}$ which is not $S^{n},$ because the Cartesian product
$\underset{n}{\underbrace{S\times...\times S}}=S^{(n)}$ is taken a little
differently, in the sense of Definitions \ref{Def6}, \ref{Def5}. In such
definitions we often identify $\mathbf{B}^{n^{\prime}}\times\mathbf{B}%
^{n^{\prime\prime}}$ with $\mathbf{B}^{n^{\prime}+n^{\prime\prime}}$ and
$S^{(n^{\prime})}\times S^{(n^{\prime\prime})}$ with $S^{(n^{\prime}%
+n^{\prime\prime})}.$
\end{remark}

\section{Asynchronous systems. Regularity}

\begin{definition}
\label{Def10}An \textbf{asynchronous system} is a multi-valued function
$f:U\rightarrow P^{\ast}(S^{(n)}),U\in P^{\ast}(S^{(m)}).\;U$ is called the
\textbf{input set} and its elements $u\in U$ are called (\textbf{admissible})
\textbf{inputs}, while the functions $x\in f(u)$ are called (\textbf{possible}%
) \textbf{states}.
\end{definition}

\begin{example}
\label{Exa1}The identity function $1_{\mathbf{B}}:\mathbf{B}\rightarrow
\mathbf{B}$ is implemented in electrical engineering by the \textbf{delay
circuit}. Such a circuit may be modelled by the system $f:S^{(1)}\rightarrow
P^{\ast}(S^{(1)}),$ called itself \textbf{delay}, which is given by the double
inequality%
\[
\underset{\xi\in\lbrack t-\tau,t)}{\bigcap}u(\xi)\leq x(t)\leq\underset{\xi
\in\lbrack t-\tau,t)}{\bigcup}u(\xi),
\]
where $\tau>0$ and $u,x\in S^{(1)}.$ We have for example%
\[
f(\chi_{\lbrack0,\infty)})=\{x|x\in S^{(1)},\underset{\xi\in\lbrack t-\tau
,t)}{\bigcap}\chi_{\lbrack0,\infty)}(\xi)\leq x(t)\leq\underset{\xi\in\lbrack
t-\tau,t)}{\bigcup}\chi_{\lbrack0,\infty)}(\xi)\}
\]%
\[
=\{x|x\in S^{(1)},\chi_{\lbrack\tau,\infty)}(t)\leq x(t)\leq\chi_{(0,\infty
)}(t)\}=\{y\cdot\chi_{(0,\tau)}\oplus\chi_{\lbrack\tau,\infty)}|y\in
S^{(1)}\},
\]
thus when the input is $\chi_{\lbrack0,\infty)},$ the state is $0$ for
$t\leq0,$ it is uncertain in the interval $(0,\tau)$ and it is equal with $1$
for $t\geq\tau.$ This system computes $\forall\lambda\in\mathbf{B}$ the value
$1_{\mathbf{B}}(\lambda)=\lambda$ in at most $\tau$ time units.
\end{example}

\begin{definition}
Let be the system $f:U\rightarrow P^{\ast}(S^{(n)}),U\in P^{\ast}(S^{(m)}).$
The function $\phi_{0}:U\rightarrow P^{\ast}(\mathbf{B}^{n})$ defined by
$\forall u\in U,$%
\[
\phi_{0}(u)=\{x(-\infty+0)|x\in f(u)\}
\]
is called the \textbf{initial state function} of $f$.
\end{definition}

\begin{definition}
\label{Def13}The binary sequence $\alpha:\mathbf{N}\rightarrow\mathbf{B}%
^{n},\alpha(k)=\alpha^{k},k\in\mathbf{N}$ is called \textbf{progressive} if
$\forall i\in\{1,...,n\},$ the set $\{k|k\in\mathbf{N},\alpha_{i}^{k}=1\}$ is
infinite. The set of the progressive sequences is denoted by $\Pi_{n}.$
\end{definition}

\begin{definition}
\label{Def14}The function $\rho:\mathbf{R}\rightarrow\mathbf{B}^{n}$ is called
\textbf{progressive} if it is of the form%
\begin{equation}
\rho(t)=\alpha^{0}\cdot\chi_{\{t_{0}\}}(t)\oplus...\oplus\alpha^{k}\cdot
\chi_{\{t_{k}\}}(t)\oplus... \label{asr1}%
\end{equation}
with $\alpha\in\Pi_{n}$ and $(t_{k})\in Seq.$ The set of the progressive
functions is denoted by $P_{n}.$
\end{definition}

\begin{definition}
\label{Def15}Let be the function $\Phi:\mathbf{B}^{n}\times\mathbf{B}%
^{m}\rightarrow\mathbf{B}^{n}.$ For any $\nu\in\mathbf{B}^{n}$, we define the
function $\Phi^{\nu}:\mathbf{B}^{n}\times\mathbf{B}^{m}\rightarrow
\mathbf{B}^{n}$ by $\forall\mu\in\mathbf{B}^{n},\forall\lambda\in
\mathbf{B}^{m},$%
\[
\Phi^{\nu}(\mu,\lambda)=(\overline{\nu_{1}}\cdot\mu_{1}\oplus\nu_{1}\cdot
\Phi_{1}(\mu,\lambda),...,\overline{\nu_{n}}\cdot\mu_{n}\oplus\nu_{n}\cdot
\Phi_{n}(\mu,\lambda)).
\]
\end{definition}

\begin{definition}
\label{Def16}For all $\rho\in P_{n}$ like in (\ref{asr1})$,$ the function
$\Phi^{\rho}:\mathbf{B}^{n}\times S^{(m)}\times\mathbf{R}\rightarrow
\mathbf{B}^{n}$ is defined by $\forall\mu\in\mathbf{B}^{n},\forall u\in
S^{(m)},\forall t\in\mathbf{R},$%
\[
\Phi^{\rho}(\mu,u,t)=\omega_{-1}\cdot\chi_{(-\infty,t_{0})}(t)\oplus\omega
_{0}\cdot\chi_{\lbrack t_{0},t_{1})}\oplus...\oplus\omega_{k}\cdot
\chi_{\lbrack t_{k},t_{k+1})}\oplus...
\]
where the family $\omega_{k}\in\mathbf{B}^{n},k\in\mathbf{N}\cup\{-1\}$ is
given by%
\[
\omega_{-1}=\mu,
\]%
\[
\omega_{k+1}=\Phi^{\alpha^{k+1}}(\omega_{k},u(t_{k+1})).
\]
\end{definition}

\begin{definition}
\label{Def17}The system $\Xi_{\Phi}:S^{(m)}\rightarrow P^{\ast}(S^{(n)})$
defined by $\forall u\in S^{(m)},$%
\[
\Xi_{\Phi}(u)=\{\Phi^{\rho}(\mu,u,\cdot)|\mu\in\mathbf{B}^{n},\rho\in P_{n}\}
\]
is called the \textbf{universal regular asynchronous system} that is generated
by the function $\Phi.$
\end{definition}

\begin{definition}
\label{Def18}A system $f:U\rightarrow P^{\ast}(S^{(n)}),U\in P^{\ast}%
(S^{(m)})$ is called \textbf{regular} if $\Phi$ exists such that $\forall u\in
U,f(u)\subset\Xi_{\Phi}(u).$ In this case the function $\Phi$ is called the
\textbf{generator function} of $f$ and we use to write $f\subset\Xi_{\Phi}.$
\end{definition}

\begin{remark}
The asynchronous systems, as defined by us at Definition \ref{Def10},
represent a very general concept and such a multi-valued function may model
nothing. The meaning of the regular asynchronous systems from Definition
\ref{Def18} is that of indicating a condition so that $f$ really models a
circuit, namely the circuit that implements the function $\Phi.$

For $\nu\in\mathbf{B}^{n}$, the function $\Phi^{\nu}$ computes in Definition
\ref{Def15} the coordinates $\Phi_{i}(\mu,\lambda),i\in\{1,...,n\}$ like this:
if $\nu_{i}=1,$ then $\Phi_{i}^{\nu}(\mu,\lambda)=\Phi_{i}(\mu,\lambda)$ thus
$\Phi_{i}(\mu,\lambda)$ is computed; if $\nu_{i}=0,$ then $\Phi_{i}^{\nu}%
(\mu,\lambda)=\mu_{i}$ thus $\Phi_{i}(\mu,\lambda)$ is not computed. The
property of progress of the function $\rho,$ Definitions \ref{Def13},
\ref{Def14} assures in $\Phi^{\rho}(\mu,u,t)$ from Definition \ref{Def16} the
fact that the generator function $\Phi$ (which is not unique in general for
some $f$) is computed in the following way: $\forall i\in\{1,...,n\},\forall
t\in\mathbf{R},\exists t_{k+1}>t$ such that $\Phi_{i}(\omega_{k},u(t_{k+1}))$
is computed. This philosophy is different from the one of Example \ref{Exa1},
where the value of the function $1_{\mathbf{B}}$ was computed in at most
$\tau$ time units.

The terminology 'universal' referring to the regular asynchronous systems
means maximal in the sense of the inclusion from Definition \ref{Def18}.
\end{remark}

\begin{theorem}
\label{The19}a) If $f\subset\Xi_{\Phi},$ then the function $\pi:\Delta
\rightarrow P^{\ast}(P_{n})$ exists,%
\begin{equation}
\Delta=\{(\mu,u)|u\in U,\mu\in\phi_{0}(u)\} \label{asr2}%
\end{equation}
such that $\forall u\in U,$%
\begin{equation}
f(u)=\{\Phi^{\rho}(\mu,u,\cdot)|\mu\in\phi_{0}(u),\rho\in\pi(\mu,u)\}.
\label{asr3}%
\end{equation}

b) If $\pi$ exists such that (\ref{asr2}), (\ref{asr3}) hold, then
$f\subset\Xi_{\Phi}.$
\end{theorem}

\begin{proof}
a) Let $u\in U$ be arbitrary. From the fact that $f\subset\Xi_{\Phi}$ and
$f(u)\neq\emptyset$ we infer that $\forall\mu\in\phi_{0}(u),$ the set
$\{\rho|\rho\in P_{n},\Phi^{\rho}(\mu,u,\cdot)\in f(u)\}$ is non-empty. We
define $\Delta$ by equation (\ref{asr2}) and $\pi:\Delta\rightarrow P^{\ast
}(P_{n})$ by $\forall(\mu,u)\in\Delta,$%
\[
\pi(\mu,u)=\{\rho|\rho\in P_{n},\Phi^{\rho}(\mu,u,\cdot)\in f(u)\}.
\]
(\ref{asr3}) is fulfilled.

b) Obvious.
\end{proof}

\begin{definition}
For the regular system $f\subset\Xi_{\Phi},$ the function $\pi$ previously
defined is called the \textbf{computation function} of $f$.
\end{definition}

\section{Parallel connection}

\begin{definition}
Consider the systems $f^{\prime}:U^{\prime}\rightarrow P^{\ast}(S^{(n^{\prime
})}),f^{\prime\prime}:U^{\prime\prime}\rightarrow P^{\ast}(S^{(n^{\prime
\prime})}),$ $U^{\prime},U^{\prime\prime}\in P^{\ast}(S^{(m)})$ with
$U^{\prime}\cap U^{\prime\prime}\neq\emptyset.$ The system $f^{\prime
}||f^{\prime\prime}:U^{\prime}\cap U^{\prime\prime}\rightarrow P^{\ast
}(S^{(n^{\prime}+n^{\prime\prime})})$ defined by%
\[
\forall u\in U^{\prime}\cap U^{\prime\prime},(f^{\prime}||f^{\prime\prime
})(u)=f^{\prime}(u)\times f^{\prime\prime}(u)
\]
is called the \textbf{parallel connection} of the systems $f^{\prime}$ and
$f^{\prime\prime}$.
\end{definition}

\begin{remark}
The parallel connection of two systems $f^{\prime}$ and $f^{\prime\prime}$ is
the system that represents $f^{\prime},f^{\prime\prime}$ acting independently
on each other under the same input $u\in U^{\prime}\cap U^{\prime\prime}$.
\end{remark}

\section{The parallel connection of the regular systems}

\begin{notation}
Let be $\Phi^{\prime}:\mathbf{B}^{n^{\prime}}\times\mathbf{B}^{m}%
\rightarrow\mathbf{B}^{n^{\prime}}$ and $\Phi^{\prime\prime}:\mathbf{B}%
^{n^{\prime\prime}}\times\mathbf{B}^{m}\rightarrow\mathbf{B}^{n^{\prime\prime
}},$ for which we denote by $\Phi^{\prime}||\Phi^{\prime\prime}:\mathbf{B}%
^{n^{\prime}+n^{\prime\prime}}\times\mathbf{B}^{m}\rightarrow\mathbf{B}%
^{n^{\prime}+n^{\prime\prime}}$ the function $\forall((\mu^{\prime}%
,\mu^{\prime\prime}),\lambda)\in\mathbf{B}^{n^{\prime}+n^{\prime\prime}}%
\times\mathbf{B}^{m},$%
\[
(\Phi^{\prime}||\Phi^{\prime\prime})((\mu^{\prime},\mu^{\prime\prime}%
),\lambda)=(\Phi^{\prime}(\mu^{\prime},\lambda),\Phi^{\prime\prime}%
(\mu^{\prime\prime},\lambda)).
\]
\end{notation}

\begin{definition}
The function $\Phi:\mathbf{B}^{n}\times\mathbf{B}^{m}\rightarrow\mathbf{B}%
^{n}$ is given. The (\textbf{Boolean partial}) \textbf{derivative} \textbf{of}
$\Phi_{i}$ \textbf{relative to} $\mu_{j},$ $i,j\in\{1,...,n\}$ is the function
$\dfrac{\partial\Phi_{i}}{\partial\mu_{j}}:\mathbf{B}^{n}\times\mathbf{B}%
^{m}\rightarrow\mathbf{B}$ given by $\forall(\mu,\lambda)\in\mathbf{B}%
^{n}\times\mathbf{B}^{m},$%
\[
\dfrac{\partial\Phi_{i}}{\partial\mu_{j}}(\mu,\lambda)=\Phi_{i}(\mu
_{1},...,\mu_{j},...,\mu_{n},\lambda)\oplus\Phi_{i}(\mu_{1},...,\overline
{\mu_{j}},...,\mu_{n},\lambda).
\]
\end{definition}

\begin{theorem}
\label{The26}We consider the functions $\Phi^{\prime}$ and $\Phi^{\prime
\prime}$. The following statements are true:

i) $\forall(\mu,\lambda)\in\mathbf{B}^{n^{\prime}+n^{\prime\prime}}%
\times\mathbf{B}^{m},\forall i\in\{1,...,n^{\prime}\},\forall j\in\{n^{\prime
}+1,...,n^{\prime}+n^{\prime\prime}\},$%
\[
(\Phi^{\prime}||\Phi^{\prime\prime})_{i}(\mu_{1},...,\mu_{n^{\prime}}%
,...,\mu_{j},...,\mu_{n^{\prime}+n^{\prime\prime}},\lambda)=
\]%
\[
=(\Phi^{\prime}||\Phi^{\prime\prime})_{i}(\mu_{1},...,\mu_{n^{\prime}%
},...,\overline{\mu_{j}},...,\mu_{n^{\prime}+n^{\prime\prime}},\lambda),
\]
$\forall(\mu,\lambda)\in\mathbf{B}^{n^{\prime}+n^{\prime\prime}}%
\times\mathbf{B}^{m},\forall i\in\{n^{\prime}+1,...,n^{\prime}+n^{\prime
\prime}\},\forall j\in\{1,...,n^{\prime}\},$%
\[
(\Phi^{\prime}||\Phi^{\prime\prime})_{i}(\mu_{1},...,\mu_{j},...,\mu
_{n^{\prime}+1},...,\mu_{n^{\prime}+n^{\prime\prime}},\lambda)=
\]%
\[
=(\Phi^{\prime}||\Phi^{\prime\prime})_{i}(\mu_{1},...,\overline{\mu_{j}%
},...,\mu_{n^{\prime}+1},...,\mu_{n^{\prime}+n^{\prime\prime}},\lambda);
\]

ii) $\forall(\mu,\lambda)\in\mathbf{B}^{n^{\prime}+n^{\prime\prime}}%
\times\mathbf{B}^{m},$%
\[
\forall i\in\{1,...,n^{\prime}\},\forall j\in\{n^{\prime}+1,...,n^{\prime
}+n^{\prime\prime}\},\dfrac{\partial(\Phi^{\prime}||\Phi^{\prime\prime})_{i}%
}{\partial\mu_{j}}(\mu,\lambda)=0,
\]%
\[
\forall i\in\{n^{\prime}+1,...,n^{\prime}+n^{\prime\prime}\},\forall
j\in\{1,...,n^{\prime}\},\dfrac{\partial(\Phi^{\prime}||\Phi^{\prime\prime
})_{i}}{\partial\mu_{j}}(\mu,\lambda)=0;
\]

iii) $\forall(\mu,\lambda)\in\mathbf{B}^{n^{\prime}+n^{\prime\prime}}%
\times\mathbf{B}^{m},$%
\begin{equation}
(\Phi^{\prime}||\Phi^{\prime\prime})_{1}(\mu,\lambda)=\Phi_{1}^{\prime}%
(\mu_{1},...,\mu_{n^{\prime}},\lambda), \label{par3}%
\end{equation}%
\[
...
\]%
\begin{equation}
(\Phi^{\prime}||\Phi^{\prime\prime})_{n^{\prime}}(\mu,\lambda)=\Phi
_{n^{\prime}}^{\prime}(\mu_{1},...,\mu_{n^{\prime}},\lambda), \label{par4}%
\end{equation}%
\begin{equation}
(\Phi^{\prime}||\Phi^{\prime\prime})_{n^{\prime}+1}(\mu,\lambda)=\Phi
_{1}^{\prime\prime}(\mu_{n^{\prime}+1},...,\mu_{n^{\prime}+n^{\prime\prime}%
},\lambda), \label{par5}%
\end{equation}%
\[
...
\]%
\begin{equation}
(\Phi^{\prime}||\Phi^{\prime\prime})_{n^{\prime}+n^{\prime\prime}}(\mu
,\lambda)=\Phi_{n^{\prime\prime}}^{\prime\prime}(\mu_{n^{\prime}+1}%
,...,\mu_{n^{\prime}+n^{\prime\prime}},\lambda). \label{par6}%
\end{equation}
\end{theorem}

\begin{proof}
$\forall(\mu,\lambda)\in\mathbf{B}^{n^{\prime}+n^{\prime\prime}}%
\times\mathbf{B}^{m},\forall i\in\{1,...,n^{\prime}\},\forall j\in\{n^{\prime
}+1,...,n^{\prime}+n^{\prime\prime}\},$%
\[
(\Phi^{\prime}||\Phi^{\prime\prime})_{i}(\mu_{1},...,\mu_{n^{\prime}}%
,...,\mu_{j},...,\mu_{n^{\prime}+n^{\prime\prime}},\lambda)=\Phi_{i}^{\prime
}(\mu_{1},...,\mu_{n^{\prime}},\lambda)=
\]%
\[
=(\Phi^{\prime}||\Phi^{\prime\prime})_{i}(\mu_{1},...,\mu_{n^{\prime}%
},...,\overline{\mu_{j}},...,\mu_{n^{\prime}+n^{\prime\prime}},\lambda)
\]
i.e. the first part of i) holds true or, equivalently:%
\[
\dfrac{\partial(\Phi^{\prime}||\Phi^{\prime\prime})_{i}}{\partial\mu_{j}}%
(\mu,\lambda)=\Phi_{i}^{\prime}(\mu_{1},...,\mu_{n^{\prime}},\lambda
)\oplus\Phi_{i}^{\prime}(\mu_{1},...,\mu_{n^{\prime}},\lambda)=0
\]
i.e. the first part of ii) holds true. When $i$ runs in $\{1,...,n^{\prime
}\},$ (\ref{par3}),...,(\ref{par4}) are true themselves.
\end{proof}

\begin{lemma}
\label{Lem1}For any $\rho^{\prime}\in P_{n^{\prime}},\rho^{\prime\prime}\in
P_{n^{\prime\prime}},$ the function $\rho^{\prime}\times\rho^{\prime\prime
}:\mathbf{R}\rightarrow\mathbf{B}^{n^{\prime}+n^{\prime\prime}}$ is
progressive and it belongs to $P_{n^{\prime}+n^{\prime\prime}}$.
\end{lemma}

\begin{proof}
We presume (without loosing the generality) that%
\begin{equation}
\rho^{\prime}(t)=\alpha^{\prime0}\cdot\chi_{\{t_{0}\}}(t)\oplus...\oplus
\alpha^{\prime k}\cdot\chi_{\{t_{k}\}}(t)\oplus... \label{par1}%
\end{equation}%
\begin{equation}
\rho^{\prime\prime}(t)=\alpha^{\prime\prime0}\cdot\chi_{\{t_{0}\}}%
(t)\oplus...\oplus\alpha^{\prime\prime k}\cdot\chi_{\{t_{k}\}}(t)\oplus...
\label{par2}%
\end{equation}
with $\alpha^{\prime}\in\Pi_{n^{\prime}},\alpha^{\prime\prime}\in
\Pi_{n^{\prime\prime}}$ and $(t_{k})\in Seq.$ We infer%
\[
(\rho^{\prime}\times\rho^{\prime\prime})(t)=(\rho^{\prime}(t),\rho
^{\prime\prime}(t))=(\alpha^{\prime0},\alpha^{\prime\prime0})\cdot
\chi_{\{t_{0}\}}(t)\oplus...\oplus(\alpha^{\prime k},\alpha^{\prime\prime
k})\cdot\chi_{\{t_{k}\}}(t)\oplus...
\]
The function $\rho^{\prime}\times\rho^{\prime\prime}$ is progressive because
$\forall i\in\{1,...,n^{\prime}\},$ the set
\[
\{k|k\in\mathbf{N},(\alpha^{\prime k},\alpha^{\prime\prime k})_{i}%
=1\}=\{k|k\in\mathbf{N},\alpha_{i}^{\prime k}=1\}
\]
is infinite and $\forall i\in\{n^{\prime}+1,...,n^{\prime}+n^{\prime\prime
}\},$ the set%
\[
\{k|k\in\mathbf{N},(\alpha^{\prime k},\alpha^{\prime\prime k})_{i}%
=1\}=\{k|k\in\mathbf{N},\alpha_{i}^{\prime\prime k}=1\}
\]
is infinite too.
\end{proof}

\begin{theorem}
\label{The27}The functions $\Phi^{\prime}:\mathbf{B}^{n^{\prime}}%
\times\mathbf{B}^{m}\rightarrow\mathbf{B}^{n^{\prime}}$, $\Phi^{\prime\prime
}:\mathbf{B}^{n^{\prime\prime}}\times\mathbf{B}^{m}\rightarrow\mathbf{B}%
^{n^{\prime\prime}}$ are given. For any $\rho^{\prime}\in P_{n^{\prime}},$
$\rho^{\prime\prime}\in P_{n^{\prime\prime}},$ $\mu^{\prime}\in\mathbf{B}%
^{n^{\prime}},$ $\mu^{\prime\prime}\in\mathbf{B}^{n^{\prime\prime}},$ $u\in
S^{(m)}$ and any $t\in\mathbf{R}$ we have%
\[
(\Phi^{\prime}||\Phi^{\prime\prime})^{\rho^{\prime}\times\rho^{\prime\prime}%
}((\mu^{\prime},\mu^{\prime\prime}),u,t)=(\Phi^{\prime\rho^{\prime}}%
(\mu^{\prime},u,t),\Phi^{\prime\prime\rho^{\prime\prime}}(\mu^{\prime\prime
},u,t)).
\]
\end{theorem}

\begin{proof}
We suppose that $\rho^{\prime},\rho^{\prime\prime}$ are like at (\ref{par1}),
(\ref{par2}) and Lemma \ref{Lem1} shows that $\rho^{\prime}\times\rho
^{\prime\prime}$ is a progressive function from $P_{n^{\prime}+n^{\prime
\prime}}$, thus the statement of the Theorem makes sense.

For $\rho^{\prime}\in P_{n^{\prime}},$ $\rho^{\prime\prime}\in P_{n^{\prime
\prime}},$ $\mu^{\prime}\in\mathbf{B}^{n^{\prime}},$ $\mu^{\prime\prime}%
\in\mathbf{B}^{n^{\prime\prime}},$ $u\in S^{(m)}$ and $t\in\mathbf{R}$ we can
write%
\[
(\Phi^{\prime}||\Phi^{\prime\prime})^{\rho^{\prime}\times\rho^{\prime\prime}%
}((\mu^{\prime},\mu^{\prime\prime}),u,t)=\omega_{-1}\cdot\chi_{(-\infty
,t_{0})}(t)
\]%
\[
\oplus\omega_{0}\cdot\chi_{\lbrack t_{0},t_{1})}(t)\oplus...\oplus\omega
_{k}\cdot\chi_{\lbrack t_{k},t_{k+1})}(t)\oplus...
\]
where $\forall k\in\mathbf{N}\cup\{-1\},\omega_{k}\in\mathbf{B}^{n^{\prime
}+n^{\prime\prime}}$ and%
\[
\omega_{-1}=(\mu^{\prime},\mu^{\prime\prime}),
\]%
\[
\omega_{k+1}=(\Phi^{\prime}||\Phi^{\prime\prime})^{(\alpha^{\prime k+1}%
,\alpha^{\prime\prime k+1})}(\omega_{k},u(t_{k+1})).
\]
We denote by $\omega_{k}^{\prime}\in\mathbf{B}^{n^{\prime}},\omega_{k}%
^{\prime\prime}\in\mathbf{B}^{n^{\prime\prime}}$ the first $n^{\prime}$
coordinates and the last $n^{\prime\prime}$ coordinates of $\omega_{k}$ and we
remark that $\forall k\in\mathbf{N}\cup\{-1\},$%
\[
(\omega_{-1}^{\prime},\omega_{-1}^{\prime\prime})=(\mu^{\prime},\mu
^{\prime\prime}),
\]%
\[
(\omega_{k+1}^{\prime},\omega_{k+1}^{\prime\prime})=(\Phi^{\prime}%
||\Phi^{\prime\prime})^{(\alpha^{\prime k+1},\alpha^{\prime\prime k+1}%
)}((\omega_{k}^{\prime},\omega_{k}^{\prime\prime}),u(t_{k+1}))=
\]%
\[
=(\Phi^{\prime\alpha^{\prime k+1}}||\Phi^{\prime\prime\alpha^{\prime\prime
k+1}})((\omega_{k}^{\prime},\omega_{k}^{\prime\prime}),u(t_{k+1}))=
\]%
\[
=(\Phi^{\prime\alpha^{\prime k+1}}(\omega_{k}^{\prime},u(t_{k+1}%
)),\Phi^{\prime\prime\alpha^{\prime\prime k+1}}(\omega_{k}^{\prime\prime
},u(t_{k+1})))
\]
thus $(\omega_{k}^{\prime}),(\omega_{k}^{\prime\prime})$ fulfill%
\[
\Phi^{\prime\rho^{\prime}}(\mu^{\prime},u,t)=\omega_{-1}^{\prime}\cdot
\chi_{(-\infty,t_{0})}(t)\oplus\omega_{0}^{\prime}\cdot\chi_{\lbrack
t_{0},t_{1})}(t)\oplus...\oplus\omega_{k}^{\prime}\cdot\chi_{\lbrack
t_{k},t_{k+1})}(t)\oplus...
\]%
\[
\Phi^{\prime\prime\rho^{\prime\prime}}(\mu^{\prime\prime},u,t)=\omega
_{-1}^{\prime\prime}\cdot\chi_{(-\infty,t_{0})}(t)\oplus\omega_{0}%
^{\prime\prime}\cdot\chi_{\lbrack t_{0},t_{1})}(t)\oplus...\oplus\omega
_{k}^{\prime\prime}\cdot\chi_{\lbrack t_{k},t_{k+1})}(t)\oplus...
\]
The conclusion is that%
\[
(\Phi^{\prime}||\Phi^{\prime\prime})^{\rho^{\prime}\times\rho^{\prime\prime}%
}((\mu^{\prime},\mu^{\prime\prime}),u,t)=
\]%
\[
=(\omega_{-1}^{\prime},\omega_{-1}^{\prime\prime})\cdot\chi_{(-\infty,t_{0}%
)}(t)\oplus(\omega_{0}^{\prime},\omega_{0}^{\prime\prime})\cdot\chi_{\lbrack
t_{0},t_{1})}(t)\oplus...\oplus(\omega_{k}^{\prime},\omega_{k}^{\prime\prime
})\cdot\chi_{\lbrack t_{k},t_{k+1})}(t)\oplus...=
\]%
\[
=(\Phi^{\prime\rho^{\prime}}(\mu^{\prime},u,t),\Phi^{\prime\prime\rho
^{\prime\prime}}(\mu^{\prime\prime},u,t)).
\]
\end{proof}

\begin{notation}
We suppose that the systems $f^{\prime}:U^{\prime}\rightarrow P^{\ast
}(S^{(n^{\prime})}),f^{\prime\prime}:U^{\prime\prime}\rightarrow P^{\ast
}(S^{(n^{\prime\prime})}),U^{\prime},U^{\prime\prime}\in P^{\ast}(S^{(m)})$
are regular i.e. $f^{\prime}\subset\Xi_{\Phi^{\prime}},f^{\prime\prime}%
\subset\Xi_{\Phi^{\prime\prime}}.$ Let $\phi_{0}^{\prime}:U^{\prime
}\rightarrow P^{\ast}(\mathbf{B}^{n^{\prime}}),$ $\phi_{0}^{\prime\prime
}:U^{\prime\prime}\rightarrow P^{\ast}(\mathbf{B}^{n^{\prime\prime}})$ be
their initial state functions and $\pi^{\prime}:\Delta^{\prime}\rightarrow
P^{\ast}(P_{n^{\prime}}),$ $\pi^{\prime\prime}:\Delta^{\prime\prime
}\rightarrow P^{\ast}(P_{n^{\prime\prime}})$ be their computation functions,%
\[
\Delta^{\prime}=\{(\mu^{\prime},u^{\prime})|u^{\prime}\in U^{\prime}%
,\mu^{\prime}\in\phi_{0}^{\prime}(u^{\prime})\},
\]%
\[
\Delta^{\prime\prime}=\{(\mu^{\prime\prime},u^{\prime\prime})|u^{\prime\prime
}\in U^{\prime\prime},\mu^{\prime\prime}\in\phi_{0}^{\prime\prime}%
(u^{\prime\prime})\}.
\]
If $U^{\prime}\cap U^{\prime\prime}\neq\emptyset,$ then we use the notations
$\phi_{||}:U^{\prime}\cap U^{\prime\prime}\rightarrow P^{\ast}(\mathbf{B}%
^{n^{\prime}+n^{\prime\prime}}),$ $\pi_{||}:\Delta^{\prime}||\Delta
^{\prime\prime}\rightarrow P^{\ast}(P_{n^{\prime}+n^{\prime\prime}})$ for the
functions $\forall u\in$ $U^{\prime}\cap U^{\prime\prime},$%
\[
\phi_{||}(u)=\phi_{0}^{\prime}(u)\times\phi_{0}^{\prime\prime}(u)
\]
and respectively%
\[
\Delta^{\prime}||\Delta^{\prime\prime}=\{((\mu^{\prime},\mu^{\prime\prime
}),u)|u\in U^{\prime}\cap U^{\prime\prime},\mu^{\prime}\in\phi_{0}^{\prime
}(u),\mu^{\prime\prime}\in\phi_{0}^{\prime\prime}(u)\},
\]
$\forall((\mu^{\prime},\mu^{\prime\prime}),u)\in\Delta^{\prime}||\Delta
^{\prime\prime},$%
\[
\pi_{||}((\mu^{\prime},\mu^{\prime\prime}),u)=\pi^{\prime}(\mu^{\prime
},u)\times\pi^{\prime\prime}(\mu^{\prime\prime},u).
\]
\end{notation}

\begin{theorem}
If $f^{\prime}\subset\Xi_{\Phi^{\prime}},$ $f^{\prime\prime}\subset\Xi
_{\Phi^{\prime\prime}}$ and $U^{\prime}\cap U^{\prime\prime}\neq\emptyset,$
then $f^{\prime}||f^{\prime\prime}\subset\Xi_{\Phi^{\prime}||\Phi
^{\prime\prime}},$ its initial state function is $\phi_{||}$ and its
computation function is $\pi_{||}.$
\end{theorem}

\begin{proof}
We prove first that the initial state function of $f^{\prime}\times
f^{\prime\prime}$ is $\phi_{||}:$ $\forall u\in U^{\prime}\cap U^{\prime
\prime},$%
\[
\{z(-\infty+0)|z\in(f^{\prime}||f^{\prime\prime})(u)\}=
\]%
\[
=\{(x^{\prime}(-\infty+0),x^{\prime\prime}(-\infty+0))|(x^{\prime}%
,x^{\prime\prime})\in f^{\prime}(u)\times f^{\prime\prime}(u)\}=
\]%
\[
=\{(x^{\prime}(-\infty+0),x^{\prime\prime}(-\infty+0))|x^{\prime}\in
f^{\prime}(u),x^{\prime\prime}\in f^{\prime\prime}(u)\}=
\]%
\[
=\{x^{\prime}(-\infty+0)|x^{\prime}\in f^{\prime}(u)\}\times\{x^{\prime\prime
}(-\infty+0)|x^{\prime\prime}\in f^{\prime\prime}(u)\}=
\]%
\[
=\phi_{0}^{\prime}(u)\times\phi_{0}^{\prime\prime}(u)=\phi_{||}(u).
\]

We infer $\forall u\in U^{\prime}\cap U^{\prime\prime},$%
\[
(f^{\prime}||f^{\prime\prime})(u)=f^{\prime}(u)\times f^{\prime\prime}(u)=
\]%
\[
=\{\Phi^{\prime\rho^{\prime}}(\mu^{\prime},u,\cdot)|\mu^{\prime}\in\phi
_{0}^{\prime}(u),\rho^{\prime}\in\pi^{\prime}(\mu^{\prime},u)\}\times
\]%
\[
\times\{\Phi^{\prime\prime\rho^{\prime\prime}}(\mu^{\prime\prime},u,\cdot
)|\mu^{\prime\prime}\in\phi_{0}^{\prime\prime}(u),\rho^{\prime\prime}\in
\pi^{\prime\prime}(\mu^{\prime\prime},u)\}=
\]%
\[
=\{(\Phi^{\prime\rho^{\prime}}(\mu^{\prime},u,\cdot),\Phi^{\prime\prime
\rho^{\prime\prime}}(\mu^{\prime\prime},u,\cdot))|
\]%
\[
\mu^{\prime}\in\phi_{0}^{\prime}(u),\mu^{\prime\prime}\in\phi_{0}%
^{\prime\prime}(u),\rho^{\prime}\in\pi^{\prime}(\mu^{\prime},u),\rho
^{\prime\prime}\in\pi^{\prime\prime}(\mu^{\prime\prime},u)\}=
\]%
\[
=\{(\Phi^{\prime\rho^{\prime}}(\mu^{\prime},u,\cdot),\Phi^{\prime\prime
\rho^{\prime\prime}}(\mu^{\prime\prime},u,\cdot))|
\]%
\[
(\mu^{\prime},\mu^{\prime\prime})\in\phi_{0}^{\prime}(u)\times\phi_{0}%
^{\prime\prime}(u),\rho^{\prime}\times\rho^{\prime\prime}\in\pi^{\prime}%
(\mu^{\prime},u)\times\pi^{\prime\prime}(\mu^{\prime\prime},u)\}=
\]%
\[
=\{(\Phi^{\prime\rho^{\prime}}(\mu^{\prime},u,\cdot),\Phi^{\prime\prime
\rho^{\prime\prime}}(\mu^{\prime\prime},u,\cdot))|
\]%
\[
(\mu^{\prime},\mu^{\prime\prime})\in\phi_{||}(u),\rho^{\prime}\times
\rho^{\prime\prime}\in\pi_{||}((\mu^{\prime},\mu^{\prime\prime}),u)\}=
\]%
\[
\overset{\text{Theorem \ref{The27}}}{=}\{(\Phi^{\prime}||\Phi^{\prime\prime
})^{\rho^{\prime}\times\rho^{\prime\prime}}((\mu^{\prime},\mu^{\prime\prime
}),u,\cdot)|
\]%
\[
(\mu^{\prime},\mu^{\prime\prime})\in\phi_{||}(u),\rho^{\prime}\times
\rho^{\prime\prime}\in\pi_{||}((\mu^{\prime},\mu^{\prime\prime}),u)\}.
\]
We apply Theorem \ref{The19} b).
\end{proof}

\section{The decomposition of the systems as parallel connection of systems}

\begin{theorem}
\label{The30}The function $\Phi$ and the numbers $n^{\prime},n^{\prime\prime
}>0,n^{\prime}+n^{\prime\prime}=n$ are given. The following statements are
equivalent (see Theorem \ref{The26}):

i) $\forall(\mu,\lambda)\in\mathbf{B}^{n}\times\mathbf{B}^{m},\forall
i\in\{1,...,n^{\prime}\},\forall j\in\{n^{\prime}+1,...,n\},$%
\[
\Phi_{i}(\mu_{1},...,\mu_{n^{\prime}},...,\mu_{j},...,\mu_{n},\lambda
)=\Phi_{i}(\mu_{1},...,\mu_{n^{\prime}},...,\overline{\mu_{j}},...,\mu
_{n},\lambda),
\]
$\forall(\mu,\lambda)\in\mathbf{B}^{n}\times\mathbf{B}^{m},\forall
i\in\{n^{\prime}+1,...,n\},\forall j\in\{1,...,n^{\prime}\},$%
\[
\Phi_{i}(\mu_{1},...,\mu_{j},...,\mu_{n^{\prime}+1},...,\mu_{n},\lambda
)=\Phi_{i}(\mu_{1},...,\overline{\mu_{j}},...,\mu_{n^{\prime}+1},...,\mu
_{n},\lambda);
\]

ii) $\forall(\mu,\lambda)\in\mathbf{B}^{n}\times\mathbf{B}^{m},$%
\[
\forall i\in\{1,...,n^{\prime}\},\forall j\in\{n^{\prime}+1,...,n\},\dfrac
{\partial\Phi_{i}}{\partial\mu_{j}}(\mu,\lambda)=0,
\]%
\[
\forall i\in\{n^{\prime}+1,...,n\},\forall j\in\{1,...,n^{\prime}%
\},\dfrac{\partial\Phi_{i}}{\partial\mu_{j}}(\mu,\lambda)=0;
\]

iii) the functions $\Phi^{\prime}:\mathbf{B}^{n^{\prime}}\times\mathbf{B}%
^{m}\rightarrow\mathbf{B}^{n^{\prime}},$ $\Phi^{\prime\prime}:\mathbf{B}%
^{n^{\prime\prime}}\times\mathbf{B}^{m}\rightarrow\mathbf{B}^{n^{\prime\prime
}}$ exist such that $\forall(\mu,\lambda)\in\mathbf{B}^{n}\times\mathbf{B}%
^{m},$%
\begin{equation}
\Phi_{1}(\mu,\lambda)=\Phi_{1}^{\prime}(\mu_{1},...,\mu_{n^{\prime}},\lambda),
\label{dec1}%
\end{equation}%
\[
...
\]%
\begin{equation}
\Phi_{n^{\prime}}(\mu,\lambda)=\Phi_{n^{\prime}}^{\prime}(\mu_{1}%
,...,\mu_{n^{\prime}},\lambda), \label{dec2}%
\end{equation}%
\begin{equation}
\Phi_{n^{\prime}+1}(\mu,\lambda)=\Phi_{1}^{\prime\prime}(\mu_{n^{\prime}%
+1},...,\mu_{n^{\prime}+n^{\prime\prime}},\lambda), \label{dec3}%
\end{equation}%
\[
...
\]%
\begin{equation}
\Phi_{n^{\prime}+n^{\prime\prime}}(\mu,\lambda)=\Phi_{n^{\prime\prime}%
}^{\prime\prime}(\mu_{n^{\prime}+1},...,\mu_{n^{\prime}+n^{\prime\prime}%
},\lambda). \label{dec4}%
\end{equation}
\end{theorem}

\begin{proof}
i)$\Longleftrightarrow$ii) is obvious.

i)$\Longrightarrow$iii) From the fact that $\forall(\mu,\lambda)\in
\mathbf{B}^{n}\times\mathbf{B}^{m},$ $\forall i\in\{1,...,n^{\prime}\},$ we
have%
\[
\Phi_{i}(\mu_{1},...,\mu_{n^{\prime}},0,...,0,0)=\Phi_{i}(\mu_{1}%
,...,\mu_{n^{\prime}},0,...,0,1)=
\]%
\[
=\Phi_{i}(\mu_{1},...,\mu_{n^{\prime}},0,...,1,0)=\Phi_{i}(\mu_{1}%
,...,\mu_{n^{\prime}},0,...,1,1)=...
\]%
\[
...=\Phi_{i}(\mu_{1},...,\mu_{n^{\prime}},1,...,1,1),
\]
we infer the existence of $\Phi^{\prime}$ such that (\ref{dec1}%
),...,(\ref{dec2}) are true.

iii)$\Longrightarrow$i) From the existence of $\Phi^{\prime}$ such that
$\forall(\mu,\lambda)\in\mathbf{B}^{n}\times\mathbf{B}^{m},$ $\forall
i\in\{1,...,n^{\prime}\},$ $\forall j\in\{n^{\prime}+1,...,n\},$%
\[
\Phi_{i}(\mu_{1},...,\mu_{n^{\prime}},...,\mu_{j},...,\mu_{n},\lambda
)=\Phi_{i}^{\prime}(\mu_{1},...,\mu_{n^{\prime}},\lambda)=
\]%
\[
=\Phi_{i}(\mu_{1},...,\mu_{n^{\prime}},...,\overline{\mu_{j}},...,\mu
_{n},\lambda),
\]
we get that the first part of i) is fulfilled.
\end{proof}

\begin{definition}
If one of the previous properties i), ii), iii) is true, we say that
\textbf{the coordinates} $\Phi_{1},...,\Phi_{n^{\prime}}$ \textbf{do not
depend on} $\mu_{n^{\prime}+1},...,\mu_{n}$ and that \textbf{the coordinates}
$\Phi_{n^{\prime}+1},...,\Phi_{n}$ \textbf{do not depend on} $\mu_{1}%
,...,\mu_{n^{\prime}}.$ The coordinates $\{1,...,n^{\prime}\}$ and
$\{n^{\prime}+1,...,n\}$ are called \textbf{separated}. We also say that
$\{1,...,n\}$ accepts the $\Phi-$\textbf{partition} $\{1,...,n^{\prime}\},$
$\{n^{\prime}+1,...,n\}.$
\end{definition}

\begin{theorem}
\label{The32}We suppose that $\Phi_{1},...,\Phi_{n^{\prime}}$ do not depend on
$\mu_{n^{\prime}+1},...,$ $\mu_{n}$ and that $\Phi_{n^{\prime}+1},...,$
$\Phi_{n}$ do not depend on $\mu_{1},...,$ $\mu_{n^{\prime}}.$ Then the
functions $\Phi^{\prime}:\mathbf{B}^{n^{\prime}}\times\mathbf{B}%
^{m}\rightarrow\mathbf{B}^{n^{\prime}},$ $\Phi^{\prime\prime}:\mathbf{B}%
^{n^{\prime\prime}}\times\mathbf{B}^{m}\rightarrow\mathbf{B}^{n^{\prime\prime
}}$ exist such that $\Phi=\Phi^{\prime}||\Phi^{\prime\prime}.$
\end{theorem}

\begin{proof}
From Theorem \ref{The30} iii) we have the existence of $\Phi^{\prime}%
,\Phi^{\prime\prime}$ such that (\ref{dec1}),..., (\ref{dec2}), (\ref{dec3}%
),..., (\ref{dec4}) are fulfilled. We denote by $\mu^{\prime},\mu
^{\prime\prime}$ the first $n^{\prime}$ coordinates of $\mu\in\mathbf{B}^{n}$
and respectively the last $n^{\prime\prime}$ coordinates of $\mu.$ We have
$\forall(\mu,\lambda)\in\mathbf{B}^{n}\times\mathbf{B}^{m},$%
\[
\Phi((\mu^{\prime},\mu^{\prime\prime}),\lambda)=(\Phi_{1}((\mu^{\prime}%
,\mu^{\prime\prime}),\lambda),...
\]%
\[
...,\Phi_{n^{\prime}}((\mu^{\prime},\mu^{\prime\prime}),\lambda),\Phi
_{n^{\prime}+1}((\mu^{\prime},\mu^{\prime\prime}),\lambda),...,\Phi_{n}%
((\mu^{\prime},\mu^{\prime\prime}),\lambda))=
\]%
\[
=(\Phi_{1}^{\prime}(\mu^{\prime},\lambda),...,\Phi_{n^{\prime}}^{\prime}%
(\mu^{\prime},\lambda),\Phi_{1}^{\prime\prime}(\mu^{\prime\prime}%
,\lambda),...,\Phi_{n^{\prime\prime}}^{\prime\prime}(\mu^{\prime\prime
},\lambda))=
\]%
\[
=(\Phi^{\prime}(\mu^{\prime},\lambda),\Phi^{\prime\prime}(\mu^{\prime\prime
},\lambda))=(\Phi^{\prime}||\Phi^{\prime\prime})((\mu^{\prime},\mu
^{\prime\prime}),\lambda).
\]
\end{proof}

\begin{notation}
Let be the system $f\subset\Xi_{\Phi},$ where $\Phi:\mathbf{B}^{n}%
\times\mathbf{B}^{m}\rightarrow\mathbf{B}^{n},f:U\rightarrow P^{\ast}%
(S^{(n)})$ and $\forall u\in U,$%
\[
f(u)=\{\Phi^{\rho}(\mu,u,\cdot)|\mu\in\phi_{0},\rho\in\pi(\mu,u)\}.
\]
We suppose that $n^{\prime},n^{\prime\prime}>0,n^{\prime}+n^{\prime\prime}=n$
are given. Then we denote by $\phi_{0}^{\prime}:U\rightarrow P^{\ast
}(\mathbf{B}^{n^{\prime}}),$ $\pi^{\prime}:\Delta^{\prime}\rightarrow P^{\ast
}(P_{n^{\prime}})$ the functions $\forall u\in U,$%
\[
\phi_{0}^{\prime}(u)=\{(\mu_{1},...,\mu_{n^{\prime}})|\mu\in\phi_{0}(u)\},
\]%
\[
\Delta^{\prime}=\{(\mu^{\prime},u)|u\in U,\mu^{\prime}\in\phi_{0}^{\prime
}(u)\}
\]
and respectively $\forall(\mu^{\prime},u)\in\Delta^{\prime},$%
\[
\pi^{\prime}(\mu^{\prime},u)=\{(\rho_{1},...,\rho_{n^{\prime}})|\exists
\mu^{\prime\prime}\in\mathbf{B}^{n^{\prime\prime}},(\mu^{\prime},\mu
^{\prime\prime})\in\phi_{0}(u)\text{,}\rho\in\pi((\mu^{\prime},\mu
^{\prime\prime}),u)\}.
\]
The functions $\phi_{0}^{\prime\prime}:U\rightarrow P^{\ast}(\mathbf{B}%
^{n^{\prime\prime}}),\pi^{\prime\prime}:\Delta^{\prime\prime}\rightarrow
P^{\ast}(P_{n^{\prime\prime}})$ are obviously defined in this moment.
\end{notation}

\begin{notation}
If $\forall u\in U,$ $\forall\mu\in\phi_{0}(u),$ $\forall\rho^{\prime}\in
\pi^{\prime}(\mu^{\prime},u),$ $\forall\rho^{\prime\prime}\in\pi^{\prime
\prime}(\mu^{\prime\prime},u),$ $\exists\widetilde{\rho}\in\pi(\mu,u),$
$\forall t\in\mathbf{R},$%
\[
\Phi^{\widetilde{\rho}}(\mu,u,t)=\Phi^{\rho^{\prime}\times\rho^{\prime\prime}%
}(\mu,u,t),
\]
then we denote $\forall u\in U,\forall\mu\in\phi_{0}(u),\pi(\mu,u)\approx
\pi^{\prime}(\mu^{\prime},u)\times\pi^{\prime\prime}(\mu^{\prime\prime},u).$
$\pi^{\prime},\pi^{\prime\prime}$ are the previous ones, $\mu^{\prime}%
,\rho^{\prime}$ are the first $n^{\prime}$ coordinates of $\mu,\rho$ and
$\mu^{\prime\prime},\rho^{\prime\prime}$ are the last $n^{\prime\prime}$
coordinates of $\mu,\rho,$ where $n^{\prime}+n^{\prime\prime}=n.$
\end{notation}

\begin{remark}
We can see that $\forall u\in U,\phi_{0}(u)\subset\phi_{0}^{\prime}%
(u)\times\phi_{0}^{\prime\prime}(u)$ and $\forall u\in U,\forall\mu\in\phi
_{0}(u),\pi(\mu,u)\subset\pi^{\prime}(\mu^{\prime},u)\times\pi^{\prime\prime
}(\mu^{\prime\prime},u)$ hold.
\end{remark}

\begin{theorem}
\label{The34}The regular system $f\subset\Xi_{\Phi}$ is given and we suppose
that the functions $\Phi^{\prime}:\mathbf{B}^{n^{\prime}}\times\mathbf{B}%
^{m}\rightarrow\mathbf{B}^{n^{\prime}},$ $\Phi^{\prime\prime}:\mathbf{B}%
^{n^{\prime\prime}}\times\mathbf{B}^{m}\rightarrow\mathbf{B}^{n^{\prime\prime
}}$ exist such that $n^{\prime},n^{\prime\prime}>0,n^{\prime}+n^{\prime\prime
}=n$ and $\forall\mu\in\mathbf{B}^{n},\forall\lambda\in\mathbf{B}^{m},$ the
equations (\ref{dec1}),..., (\ref{dec2}), (\ref{dec3}),..., (\ref{dec4}) are
fulfilled. Then $\Phi=\Phi^{\prime}||\Phi^{\prime\prime}$ and the systems
$f^{\prime}\subset\Xi_{\Phi^{\prime}},f^{\prime\prime}\subset\Xi_{\Phi
^{\prime\prime}}$ defined by $f^{\prime}:U\rightarrow P^{\ast}(S^{(n^{\prime
})}),$ $f^{\prime\prime}:U\rightarrow P^{\ast}(S^{(n^{\prime\prime})}),\forall
u\in U,$%
\[
f^{\prime}(u)=\{\Phi^{\prime\rho^{\prime}}(\mu^{\prime},u,\cdot)|\mu^{\prime
}\in\phi_{0}^{\prime}(u),\rho^{\prime}\in\pi^{\prime}(\mu^{\prime},u)\},
\]%
\[
f^{\prime\prime}(u)=\{\Phi^{\prime\prime\rho^{\prime\prime}}(\mu^{\prime
\prime},u,\cdot)|\mu^{\prime\prime}\in\phi_{0}^{\prime\prime}(u),\rho
^{\prime\prime}\in\pi^{\prime\prime}(\mu^{\prime\prime},u)\}
\]
satisfy $f\subset f^{\prime}||f^{\prime\prime}$. If $\forall u\in U,\phi
_{0}(u)=\phi_{0}^{\prime}(u)\times\phi_{0}^{\prime\prime}(u)$ and $\forall
u\in U,\forall\mu\in\phi_{0}(u),\pi(\mu,u)\approx\pi^{\prime}(\mu^{\prime
},u)\times\pi^{\prime\prime}(\mu^{\prime\prime},u),$ then we have
$f=f^{\prime}||f^{\prime\prime}.$
\end{theorem}

\begin{proof}
The fact that $\Phi=\Phi^{\prime}||\Phi^{\prime\prime}$ results from Theorem
\ref{The32}. We denote like previously with $\mu^{\prime},\mu^{\prime\prime}$
the first $n^{\prime}$ coordinates of $\mu\in\mathbf{B}^{n}$ and the last
$n^{\prime\prime}$ coordinates of $\mu$ and the notations are similar for
$\rho^{\prime},\rho^{\prime\prime}$ and $\rho\in P_{n}.$ We have $\forall u\in
U,$%
\[
f(u)=\{\Phi^{\rho}(\mu,u,\cdot)|\mu\in\phi_{0}(u),\rho\in\pi(\mu,u)\}=
\]%
\[
=\{(\Phi^{\prime}||\Phi^{\prime\prime})^{\rho^{\prime}\times\rho^{\prime
\prime}}((\mu^{\prime},\mu^{\prime\prime}),u,\cdot)|\mu\in\phi_{0}(u),\rho
\in\pi(\mu,u)\}=
\]%
\[
\overset{\text{Theorem \ref{The27}}}{=}\{(\Phi^{\prime\rho^{\prime}}%
(\mu^{\prime},u,\cdot),\Phi^{\prime\prime\rho^{\prime\prime}}(\mu
^{\prime\prime},u,\cdot))|\mu\in\phi_{0}(u),\rho\in\pi(\mu,u)\}\subset
\]%
\[
\subset\{(\Phi^{\prime\rho^{\prime}}(\mu^{\prime},u,\cdot),\Phi^{\prime
\prime\rho^{\prime\prime}}(\mu^{\prime\prime},u,\cdot))|
\]%
\[
|(\mu^{\prime},\mu^{\prime\prime})\in\phi_{0}^{\prime}(u)\times\phi
_{0}^{\prime\prime}(u),\rho^{\prime}\times\rho^{\prime\prime}\in\pi^{\prime
}(\mu^{\prime},u)\times\pi^{\prime\prime}(\mu^{\prime\prime},u)\}=
\]%
\[
=\{\Phi^{\prime\rho^{\prime}}(\mu^{\prime},u,\cdot)|\mu^{\prime}\in\phi
_{0}^{\prime}(u),\rho^{\prime}\in\pi^{\prime}(\mu^{\prime},u)\}\times
\]%
\[
\times\{\Phi^{\prime\prime\rho^{\prime\prime}}(\mu^{\prime\prime},u,\cdot
)|\mu^{\prime\prime}\in\phi_{0}^{\prime\prime}(u),\rho^{\prime\prime}\in
\pi^{\prime\prime}(\mu^{\prime\prime},u)\}=
\]%
\[
=f^{\prime}(u)\times f^{\prime\prime}(u)=(f^{\prime}||f^{\prime\prime})(u).
\]
The second statement of the theorem is obvious.
\end{proof}

\end{document}